\documentclass[letterpaper]{article}
\usepackage{amsmath}
\usepackage{amsfonts}
\usepackage{amssymb}
\usepackage{amsthm}
\usepackage{graphicx}

\numberwithin{equation}{section}
\newtheorem{theorem}{Theorem}
\newtheorem{lemma}{Lemma}

\newtheorem{assumption}{Assumption}
\theoremstyle{remark}
\newtheorem{rmark}{Remark}

\newcommand\RE{\mathbb{R}}
\newcommand\rd{\mathrm{d}}
\newcommand\ru{\mathrm{u}}
\newcommand\rv{\mathrm{v}}

\DeclareMathOperator{\graph}{Graph}
\DeclareMathOperator{\convex}{conv}
\DeclareMathOperator{\Proj}{Proj}
\DeclareMathOperator{\Dom}{Dom}
\DeclareMathOperator{\interior}{int}
\DeclareMathOperator{\sign}{sign}

\DeclareMathOperator*{\Argmin}{arg\,min}

\title{Robust Output Regulation of \\ 
 Linear Passive Systems with \\
 Multivalued Upper Semicontinuous Controls}
\author{F\'elix A. Miranda and Fernando Casta\~nos}

\begin{document}
\maketitle

\begin{abstract}
The use of multivalued controls derived from a special maximal
monotone operator are studied in this note. Starting with a 
strictly passive linear system (with possible parametric 
uncertainty and external disturbances) a multivalued control 
law is derived, ensuring regulation of the output to a 
desired value.
The methodology used falls in a passivity-based control context,
where we study how the multivalued control affects the dissipation
equation of the closed-loop system, from which we derive its
robustness properties. Finally, some numerical examples together with
implementation issues are presented to support the main result.

\end{abstract}

\section{Introduction}


Sometimes it is useful to have an interpretation of the
action of the controller in energetic terms. Among the 
most important methodologies of passivity-based control (PBC) 
that achieve this interpretation are the so-called energy shaping
techniques. The purpose of energy shaping, as its name suggests, 
is to change the energy function (by means of the control action) in
such a way that stabilization and performance objectives are
satisfied. Although energy-shaping strategies have proved to be very
useful yielding an easy interpretation of the controller in energetic 
terms~\cite{ortega1998}, robustness against external perturbations
and model uncertainty is still a topic of research.

On the other hand, the study of differential inclusions for modelling
and analysis of processes in control theory is extensive (e.g.~\cite{acary2008,filippov1988,leine2008}),
whereas the problem of designing a multivalued control in order to achieve a
desired response is less explored, except for the case of sliding 
mode control, which takes advantage of the multivalued nature of the
signum multifunction in order to ensure robustness of the closed-loop
system.

An important family of differential inclusions (more general than
those obtained by using sliding modes techniques) are those for 
which its right-hand side is represented by maximal monotone
operators. 
In the case of linear plants, the closed-loop system is sometimes
called a multivalued Lur'e dynamical system, for which results about
existence and uniqueness of solutions have been proved in~\cite{brogliato2004,brogliato2011,brogliato2013,camlibel2002}.
This kind of systems are related to complementarity and projected
dynamical systems~\cite{brogliato2006}, which makes its study important for a 
broad range of applications coming from different fields such
as automatic control, economics, mechanics, etc.


The main contribution of this note consists in a design procedure for
a multivalued-control --- where the multivalued part is represented
by the subdifferential of some proper, convex and lower
semicontinuous function --- which achieves finite-time regulation of
the desired output together with insensitivity against a family of
bounded and unmatched perturbations.

The proposed multivalued control strategy differs remarkably
from those which are common in Sliding-Mode Control in the sense
that we obtain finite-time regulation and disturbance rejection
without a discontinuous right-hand side and therefore without the 
necessity of solutions of the associated system in the the
sense of Filippov.

This note is organized as follows. In Section~\ref{sec:ProblemForm}
the class of systems that we consider is established in conjunction
with the class of perturbations that it will be treated. The 
multivalued structure of the controller is presented and 
well-posedness of the closed-loop system is established. In Section~\ref{sec:Main},
we introduce the main result of this note. Namely, 
robustness and finite time convergence of the closed-loop system
are demonstrated. Section~\ref{sec:Implementation} touches the point 
about implementation of the multivalued control law by introducing
a regularization of the multivalued map. Some examples are presented
showing the closed loop properties. The note ends with some
conclusions and future research lines in Section~\ref{sec:Conclusions}.


\subsection{Notations and some basics}

Throughout this note, all vectors are column
vectors, even the gradient of a scalar function that we denote 
by $\nabla H(x) = \frac{\partial H(x)}{\partial x}$. A matrix 
$A \in \RE^{n \times n}$ is called positive definite (denoted 
as $A > 0$), if $w^{\top} A w > 0  $ for all $w \in \RE^{n} 
\setminus \{0\}$ (note that we are not assuming $A$ symmetric). 

A set-valued function or multifunction $\mathbf{F}: \RE^{n}
\to 2^{\RE^{n}}$ is a map that associates with any $w \in
\RE^{n}$ a subset $\mathbf{F}(w) \subset \RE^{n}$. The 
domain of $\mathbf{F}$ is given by
\begin{displaymath}
\Dom \mathbf{F} = \{ w \in \RE^{n}: \mathbf{F}(w) \neq \emptyset\} \;,
\end{displaymath}
related with the definition of a multifunction is the concept 
of its graph,
\begin{displaymath}
\graph \mathbf{F} = \{ (w,z) \in \RE^{n} \times \RE^{n}: 
z \in \mathbf{F}(w) \} \;.
\end{displaymath}
The graph is used to define the concept of monotonicity of
a multifunction in the following way: A set-valued function
$\mathbf{F}$ is said to be monotone if for all $(w,z) \in \graph 
\mathbf{F}$ and all $(w',z') \in \graph \mathbf{F}$ the relation 
\begin{displaymath}
\langle z - z', w - w' \rangle \geq 0
\end{displaymath}
is preserved, where $\langle \cdot, \cdot \rangle $ denotes the usual
scalar product on $ \RE^{n} $. A monotone map $\mathbf{F}$ is called
maximal monotone if, for every pair $(\hat{w}, \hat{z}) \in \RE^{n}
\times \RE^{n} \setminus \graph \mathbf{F}$, there exits $(w,z) \in
\graph \mathbf{F}$ with $\langle z - \hat{z}, w - \hat{w} \rangle <
0$, or in other words, if no enlargement of its graph is possible in 
$\RE^{n} \times \RE^{n}$ without destroying monotonicity.

Let $f: \RE^{n} \to \RE \cup \{+ \infty\}$
be a proper, convex and lower semi-continuous function.
The effective domain of $f$ is given by
\begin{displaymath}
\Dom f = \{w \in \RE^{n} : f(w) < \infty \} \;.
\end{displaymath}
We say that $f$ is proper if its effective domain is non empty.
The subdifferential $\partial f(w)$ of $f(\cdot)$ at $ w \in
\RE^{n} $ is defined by
\begin{displaymath}
\partial f(w) = \{ \zeta \in \RE^{n} : f(\sigma) - f(w)
\geq \langle \zeta , \sigma - w \rangle \text{ for all } \sigma \in
\RE^{n} \} \;.
\end{displaymath}
An important convex function is the indicator function of a convex
set $S$, defined by
\begin{displaymath}
 \psi_{\mathcal{S}}(w) =
  \begin{cases}
         0 & \mbox{if } w \in S \\
   +\infty & \mbox{if } w \notin S
\end{cases} \;.
\end{displaymath}
It is easy to see that when $f(\cdot)$ is equal to the indicator
function of a closed convex set $S$, then the subdifferential 
coincides with the normal cone of the set $S$ at the point 
$w \in S$, i.e.,
\begin{displaymath}
\partial \psi_{S}(w) = N_{S}(w) = \{ \xi \in \RE^{n} : 0 \geq 
\langle \xi , \sigma - w \rangle \text{ for all } \sigma \in S \} \;.
\end{displaymath}
Note that if $w$ is in the interior of $S$ then $N_{S}(w) = \{ 0 \}$.
If $w \notin S$ then $N_{S}(w) = \emptyset$.

\section{The output regulation problem} \label{sec:ProblemForm}

Consider the following affine system:
\begin{equation} \label{eq:OriginalSys}
\Sigma : \begin{cases}
\dot{x}(t)  =  Ax(t) + B_{\ru} u_1(t) + B_{\rv} v(t) \cr
y_1(t)  =  C x(t) + D u_1(t)
\end{cases} \;,
\end{equation}
where $x \in \RE^{n}$ denotes the system state,
$u_1, y_1 \in \RE^{m}$ are the port variables available for 
interconnection, which are conjugated in the sense that their product
has units of power, and matrices $A, B_{\ru}, B_{\rv}, C, D$ 
are constant and of suitable dimensions. The term 
$v \in \RE^{m}$ accounts for an uncertain exogenous input 
which is considered bounded. Moreover, without loss of generality,
the external signal $v(t)$ can be decomposed as the sum of a
constant term $v^{+} $ and a bounded signal $\nu(t)$.  

%
%
The robust output regulation problem consists in regulating the 
output $y_1$ to a desired value $y_{\rd}$, even in the presence of
the external perturbation $v(t)$ and parametric uncertainties.

\begin{rmark}
Notice that, for $D = 0$ and $B_{\ru} = B_{\rv}$ , the
problem reduces to a standard sliding-mode control problem 
with matched disturbances. We depart from these standard 
assumptions and make the following instead.
\end{rmark}
\begin{assumption} \label{ass:pass}
There exists a (possibly unknown) matrix $P = P^\top > 0$ 
such that
\begin{equation} \label{eq:LMI1}
 \begin{bmatrix}
  P A + A^{\top} P & P B_\ru - C^{\top} \\  
  B_\ru^{\top} P - C & -(D + D^{\top}) \end{bmatrix} < 0 \;.
\end{equation}
\end{assumption}
Assumption~\ref{ass:pass} is a rewrite of the strict
passivity property of plant~\eqref{eq:OriginalSys} with 
respect to the input $u_1$ and output $y_1$~\cite{knockaert2005}.
Moreover, is easy to see that the
passivity assumption implies that $D$ is positive definite.
Equivalently, Assumption~\ref{ass:pass} can be rewritten in terms
of the energy-balance equation. More
precisely, there exists a continuously differentiable function 
$H: \RE^{n} \to \RE$, called the storage function, such that for 
all $t \geq 0$ we have
\begin{displaymath}
H(x(t)) - H(x(0)) = \int_{0}^{t} u_1^{\top}(\tau) y_1(\tau) d\tau -
 \rd(t) \;,
\end{displaymath}
where the function $\rd(t) \geq 0$ is associated to the dissipation of the
system and the storage function $H(x)$ is bounded from below (see, e.g.,~\cite{ortega1998}).

\begin{rmark}
The strict passivity assumption allows us to admit the quadratic 
function $H_0(x) = x^{\top} P x$ as a storage function with $P$ 
satisfying~\eqref{eq:LMI1}.
\end{rmark}

It is worth noting that in the linear case, the class of passive 
systems is equivalent to the class of Port-Hamiltonian (PH) systems
described in~\cite[Ch. 4]{shaft1996}, i.e., $\Sigma$ can be
written as
\begin{align*}
\dot{x} &= F \nabla H_0(x) + g_{\ru } u_1 + g_{\rv} v \\
y_1 & = h(x) + j u_1
\end{align*}
with $F = A P^{-1} = J - R$, where $J = -J^{\top}$ and $R = R^{\top}
\geq 0$ are the so called interconnection and dissipation matrices,
respectively, $g_{\ru} = B_{\ru}$, $g_{\rv} = B_{\rv}$, $h(x) = Cx$
and $j = D$.

Along this note we will use both representations
of $\Sigma$ with the purpose of expressing the related
computations in the context of basic interconnection and damping
assignment (IDA)~\cite{ortega2002,shaft1996}.


\subsection{Multivalued control law}

In this subsection a multivalued control law is introduced
by using maximal monotone operators. It will be shown later on
that these are robust in the face of parametric and additive
uncertainties.

Let $u_2 \in \RE^{m}$ and $y_2 \in \RE^{m}$ be the 
port variables available for interconnection associated to the
controller. The multivalued control input is defined in terms of the
graph of a multifunction $\mathbf{U}:\RE^{m} \to 2^{\RE^{m}}$ by
\begin{displaymath}
(u_2, y_2) \in \graph \mathbf{U} \;.
\end{displaymath}
\begin{rmark}
It is worthy to mention that, in the case when the multifunction
$\mathbf{U}$ is monotone, the relation $(u_2,y_2) \in \graph 
\mathbf{U}$ defines a static, incrementally passive map\footnote{A
multivalued map $\mathbf{F}$ is called incrementally passive if 
$\langle y - y', u - u' \rangle \geq 0$ for all $(u, y) \in \graph
\mathbf{F}$ and for all $(u', y') \in \graph \mathbf{F}$.}.
Furthermore, if $0 \in \mathbf{U}(0)$, then the relation between 
$u_2$ and $y_2$ defines a static passive map inasmuch as
\begin{displaymath}
\langle u_2, y_2 \rangle \geq 0 \text{ for all } (u_2, y_2) \in 
\graph \mathbf{U} \;.
\end{displaymath} 

\end{rmark}
Previous lines motivate the following assumption.
\begin{assumption}\label{ass:mono}
The multifunction $\mathbf{U}$ is maximal monotone, and defines a
static passive relation between the input $u_2$ and the output $y_2$. 
\end{assumption}
%
%
The multivalued nature of the proposed control motivates us to
depart from the classical intelligent control paradigm and to make 
use of the behavioural framework proposed by Willems~\cite{willems1997}
instead. In this context, the plant and the controller are interconnected
using a power preserving pattern as
shown in Figure~\ref{fig:int} satisfying: $y_1 = y_2 =: y $, $-u_1 =
u_2 =: u$ and therefore
\begin{displaymath}
u_1 y_1 + u_2 y_2 = 0 \;.
\end{displaymath}
\begin{figure}[!tb]
\begin{center}
 \includegraphics[width=0.35\textwidth]{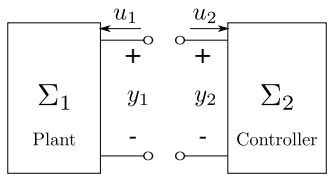}
 \caption{Interconnection of a controller to a plant.}
 \label{fig:int}
\end{center}
\end{figure}

The interconnected system (plant and controller) results in
\begin{subequations} \label{eq:SystemInterconnected}
\begin{align}
 \dot{x} &= A x - B_\ru u + B_\rv v \label{eq:ODEI} \\
     y &= C x - D u \label{eq:outputI} \\
     u &\in \mathbf{U}(y) \;, \label{eq:U}
\end{align}
\end{subequations}
where our task is to determine $\mathbf{U}(y)$ such that
$y$ is regulated to some fixed value $y_{\rd}$, 
even in the presence of uncertainties in the system parameters
and the external perturbation $v$. Note that the previous argument
rules out the trivial control $u = D^{-1}(C x - y_{\rd})$. In fact,
even if all the system parameters and the state $x$ were known,
that control would not be admissible, since it is not passive
(see Assumption~\ref{ass:mono}).

It is well known that when $\mathbf{U}(y)$ is given as the
subdifferential of a proper, convex and lower semicontinuous
function $\Phi(\cdot)$ (i.e. $\mathbf{U}(y) = \partial \Phi(y)$),
it is a maximal monotone operator~\cite[Cor. 31.5.2]{rockafellar1970}.
Therefore, we will focus on controls of the form
\begin{equation} \label{eq:uMulti}
 u \in \partial \Phi(y) \;, 
\end{equation}
for some proper, convex and lower semicontinuous function
$\Phi: \RE^{m} \to \RE^{m}$. More specifically, in Section~\ref{sec:Main}
we will prove that, for some closed convex set $S$, robust regulation
of the output $y$ is obtained for the case when $\Phi(y) = (\varphi + \psi_{S}) (y)$, 
where $\varphi(\cdot)$ is proper, convex and lower semicontinuous
with effective domain containing $S$ and $\psi_{S}(\cdot)$ is the
indicator function of the set $S$. In other words, $\Phi$ is the
restriction of $\varphi$ to $S$.


\subsection{Well-posedness} \label{subsec:Well}
Before presenting the main result of this note about robustness
of the closed-loop system~\eqref{eq:SystemInterconnected}, is 
important first to establish its well-posedness.
Specifically, well-posedness of the closed-loop system
comprises two issues. The first question is: Is there always a 
control input $u \in \partial \Phi(y)$? 
and the second one is about uniqueness and existence of solutions 
of the associated differential inclusion~\eqref{eq:SystemInterconnected}.

For the second issue about a solution of the differential 
inclusion~\eqref{eq:SystemInterconnected}, well-posedness 
was proved previously in~\cite{brogliato2011,brogliato2013},
where the subdifferential of the conjugate function of $\Phi(y)$
together with passivity of the associated system plays a crucial role.

The first issue deserves more explanation. At first we need $y \in S$
for all time $t$, this comes from the definition of the 
subdifferential, i.e., $u \in \partial \Phi(y)$ is equivalent to
\begin{displaymath}
\Phi(\sigma) - \Phi(y) \geq \langle u, \sigma - y \rangle
\text{ for all } \sigma \in \RE^{m} \;,
\end{displaymath}
where, in the case of $\Phi(y) = (\varphi + \psi_S)(y)$ we have
\begin{equation} \label{eq:subdiffIneq}
\varphi(\sigma) - \varphi(y) + \psi_S(\sigma) - \psi_S(y)
\geq \langle u, \sigma - y \rangle \text{ for all } \sigma \in \RE^{m} \;,
\end{equation}
and it is clear that if $y \notin S$, then we will have $\partial 
\Phi(y) = \emptyset$. Then, we must guarantee that, no matter what
the initial conditions are, it is possible to find an output $y \in S$,
such that $u \in \partial \Phi(y)$ is well defined.

In the case where $\varphi \equiv 0$ and the matrix $D$ 
is symmetric, well-posedness is easy to show. Since $u \in
\partial \psi_S(y) = N_S(y)$, from the definition of normal
cone we have
\begin{displaymath}
0 \geq \langle u, \sigma - y \rangle \text{ for all } \sigma \in S \;,
\end{displaymath}
which in view of~\eqref{eq:outputI} translates to
\begin{displaymath}
 0 \geq \langle D^{-1} (Cx - y), \sigma - y \rangle = (\sigma - y)^{\top}D^{-1}(Cx - y) = \langle Cx - y, \sigma - y \rangle_{D^{-1}}
\end{displaymath}<++>
for all $\sigma \in S$, with the inner product weighted by $D^{-1}$. From~\cite[p. 117]{hiriart1993}
we have that the above inequality is the characterization of the projection of $Cx$ onto the set $S$
with the induced norm $\| \cdot \|_{D^{-1}}$, i.e.
\begin{equation} \label{eq:yProj}
y = \Proj_{S}^{D^{-1}} (Cx) = \Argmin_{\sigma \in S} 
\| Cx - \sigma \|_{D^{-1}} \;,
\end{equation}
and the control input $u$ transforms into
\begin{equation} \label{eq:ctrlPhiZero}
u = D^{-1} \left( C x - \Proj_S^{D^{-1}} (C x ) \right) \in
N_S(y) \;.
\end{equation}
Therefore, in the case of $D$ symmetric, we can find an expression 
for the output $y$ in terms of the projection operator
$\Proj_S^{D^{-1}}( \cdot )$ (note that this implies $y \in S$
independently of the state $x$).
Moreover, due to the Lipschitzian property of the projection operator~\cite[p. 118]{hiriart1993},
substitution of $u$ in~\eqref{eq:SystemInterconnected} leads to a well-posed ordinary
differential equation (not a differential inclusion!) with a
Lipchitzian right-hand side (see~\cite{miranda2014} for a detailed development in the scalar case).

For the general case where $\varphi$ is not the zero function, and 
removing the assumption about the symmetry of $D$, from~\eqref{eq:subdiffIneq}
we have that the problem consists in finding $y \in S$ such that
\begin{equation} \label{eq:hemivi}
0 \leq \langle D^{-1} y - D^{-1} C x, \sigma - y \rangle +
\varphi(\sigma) - \varphi(y) \text{ for all } \sigma \in S \;,
\end{equation}
where we made use of~\eqref{eq:outputI}.  The inequality~\eqref{eq:hemivi} is an
hemivariational inequality\footnote{The interested reader is referred
to~\cite{facchinei2003,goeleven2003,nagurney1995}, and references therein for
more information and properties about variational and hemivariational
inequalities.}, for which existence and uniqueness of solutions can be
deduced from $D^{-1}$ as the following Lemma extracted from~\cite{goeleven2003} shows.
\begin{lemma}[Lemma 5.2.1, \cite{goeleven2003}] \label{lmma:goeleven}
Suppose that $F(y)$ is continuous and strongly monotone, i.e.
\begin{displaymath}
\langle F(y) - F(y'), y - y' \rangle \geq \eta \| y - y' \|^{2}
\end{displaymath}
for all $y, y' \in \RE^{m}$ and some $\eta > 0$. Then, for each $g \in \RE^{n}$ there exists an
unique solution $y^{*} \in S$ to
\begin{displaymath}
\langle F(y) - g, \sigma - y \rangle + \Phi(\sigma) - \Phi(y)
\geq 0 \text{ for all } \sigma \in S \;.
\end{displaymath}
\end{lemma}
Since $D$ is positive definite, it is straightforward to see that
$D^{-1}$ is positive definite too. Furthermore, the linear map 
$y \mapsto D^{-1} y$ is strongly monotone (as a consequence of applying 
Rayleigh's inequality). Then using Lemma~\ref{lmma:goeleven} we have 
that the hemivariational inequality problem~\eqref{eq:hemivi} 
has a unique solution for each state $x$. In other words:
For all $x \in \RE^{n}$, there always exists a unique 
$y \in S$ such that the control $u \in \partial \Phi(y)$ is well defined.
\begin{rmark} \label{rmark:partition}
The computation of the control input which forces $y \in S$
obviously depends on the solution of the hemivariational
inequality~\eqref{eq:hemivi} and therefore it depends implicitly on the actual
state of the plant $x$. This dependency of the state, induces a partition
in the phase space. 
\end{rmark}
\begin{rmark} \label{rmark:equivCtrl}
For the case when $\varphi \equiv 0$ and $D$ is symmetric,  
we might be tempted to use~\eqref{eq:ctrlPhiZero} as control input
(because it is passive), but unfortunately it is not implementable
in our setting because it depends explicitly on the system parameters and state. 
The role of~\eqref{eq:ctrlPhiZero} is analogous to the role 
of the equivalent control in sliding modes~\cite{utkin1992},
in the sense that it is not implementable but helps to determine 
the dynamics associated to the closed-loop system. See~\cite{miranda2014}
for an example of the use of the control~\eqref{eq:ctrlPhiZero} in the
scalar case and some implementation issues.
\end{rmark}
Following the steps in~\cite{miranda2014}, the control that results
from the solution of the hemivariational inequality~\eqref{eq:hemivi}
will act as an equivalent control, in the sense of Remark~\ref{rmark:equivCtrl}
and is not implementable under the assumption
that the state and the plant parameters are unknown. The implemented
control is described in Section~\ref{sec:Implementation}.


\section{Finite-time perfect output regulation} \label{sec:Main}

The main result of this note is presented in this section. Namely,  
from an energy-shaping point of view, we show that the multivalued 
control~\eqref{eq:uMulti} can be expressed as a basic IDA controller
plus a robustifying term denoted by $\eta$, affecting directly the 
dissipation of the closed-loop system and yielding to the output 
regulation despite the presence of external and parametric 
disturbances.

From the closed-loop equation of the system, we have
\begin{subequations} \label{eq:closed-loop}
\begin{align}
 \dot{x} &= A x - B_\ru u + B_\rv v \label{eq:ODEII} \\
     y &= C x - D u \label{eq:outputII} \\
     u &\in \partial \Phi (y) \label{eq:UII}
\end{align}
\end{subequations}
with $\Phi(\cdot) = (\varphi + \psi_S)(\cdot)$ for some convex
set $S$. The perturbation input $v(t)$, decomposed as a
constant term $v^{+}$ and a bounded unknown signal $\nu(t)$, affects
the dissipation equation in the following way.

Let $\bar{x}$ be the \emph{equilibrium} point of~\eqref{eq:closed-loop}
associated to a constant perturbation ($\nu(t) \equiv 0$) and 
input $u = 0$, i.e.,
\begin{equation} \label{eq:xBarDef}
0 = A \bar{x} + B_{\rv} v^{+} \;,
\end{equation}
and let $H_0$ be the storage function of system~\eqref{eq:closed-loop}
(i.e. $H_0(x) = x^{\top} P x$ with $P$ satisfying~\eqref{eq:LMI1}).
We obtain
\begin{align*}
0 &= F \nabla H_0(\bar{x}) + B_{\rv} v^{+} \\
  &= F \nabla H_0(\bar{x}) \pm F \nabla H_0(x) + B_{\rv} v^{+} \\
  &= -F \left( \nabla H_1(x) - \nabla H_0(x)  \right) + B_{\rv} v^{+}
\end{align*}
with $H_1(x) = (x-\bar{x})^{\top} P (x - \bar{x})$. Now, defining 
$H_{a}(x) = H_1(x)- H_0(x)$ we have the basic IDA controller 
equation~\cite{castanos2009} for $v^{+}$ as
\begin{displaymath}
F \nabla H_{a}(x) = B_{\rv} v^{+} \;.
\end{displaymath}
%
%
%
Then, we have that the term $v^{+}$ acts as an energy-shaping
\emph{control} changing the storage function of the uncontrolled
system $H_0$ to $H_1$ and therefore changing the equilibrium of the
system. The closed-loop system results in
\begin{subequations} \label{eq:sysH1}
\begin{align}
\dot{x} &= F \nabla H_1(x) - B_{\ru} u + B_{\rv} \nu
\label{eq:sysH1:Dyn}\\
y &= Cx - Du \label{eq:sysH1:Output}\\
u &\in \partial \Phi(y) \label{eq:sysH1:Ctrl}
\end{align}
\end{subequations}
%


For the case $\nu = 0$, a control input 
can be designed in order to obtain the asymptotic regulation of the 
output $y$ to $y_{\rd}$ using an energy-shaping interpretation
as follows.

%
\begin{lemma} \label{lmma:ubIDA}
For system~\eqref{eq:sysH1}, let $x_{*}$ be an admissible 
equilibrium associated to the constant control 
$\bar{u} = D^{-1}(C x_{*} - y_{\rd}) $, i.e. $x_{*}$ satisfies
\begin{equation} \label{eq:equiPoint}
0 = A x_{*} - B_{\ru} D^{-1} (C x_{*} - y_{\rd}) + B_{\rv} v^{+} \;.
\end{equation}
Then, $\bar{u}$ achieves regulation of the output to $y_{\rd}$ 
when $\nu = 0$. Furthermore, $\bar{u}$ is a basic IDA controller 
and satisfies
\begin{displaymath}
F \nabla H_b(x) = -B_{\ru} \bar{u} \;,
\end{displaymath}
with $H_b(x) = H_2(x) - H_1(x)$ and $H_2(x) = (x-x_{*})^{\top} 
P (x - x_{*})$.
\end{lemma}

\begin{proof}
Let $x_{*}$ be an equilibrium of system~\eqref{eq:sysH1} satisfying~\eqref{eq:equiPoint}.
Then, from~\eqref{eq:xBarDef} we have that
\begin{displaymath}
0 = A(x_{*} - \bar{x}) -  B_{\ru} D^{-1} (C x_{*} - y_{\rd}) \;,
\end{displaymath}
or, in terms of the storage functions $H_1$ and $H_2$,
\begin{displaymath}
0 = - F(\nabla H_2(x) - \nabla H_1(x)) - B_{\ru} D^{-1} 
(C x_{*} - y_{\rd}) \;.
\end{displaymath}
Therefore, we obtain a change in the storage function from 
$H_1(x)$ with minimum at $\bar{x}$ to $H_2(x)$ with minimum 
at $x_{*}$ which implies convergence of the state $x$ to $x_{*}$.
Also, for $u = \bar{u}$ in~\eqref{eq:sysH1:Output} we have
\begin{displaymath}
y = C x - D D^{-1}(C x_{*} - y_d) = C (x - x_{*}) + y_d
\end{displaymath}
and $y \to y_d$ as $x \to x_{*}$.

\end{proof}

The control $\bar{u}$ described in Lemma~\ref{lmma:ubIDA}  shapes the
energy by changing the storage function. For the new storage
function $H_2$ we have that the control input $u = \bar{u} + \eta$ 
establishes a new dissipation equation as
\begin{equation} \label{eq:eb}
\begin{split}
 \dot{H}_2(x) &= \nabla H_2(x)^{\top} F \nabla H_2(x) - 
                  \nabla H_2(x)^{\top} B_{\ru} \eta + \nabla H_2(x)^{\top} B_{\rv} \nu \\
              &= \dfrac{1}{2} (x - x_{*})^{\top} (A^{\top} P + PA) (x -x_{*}) - (x - x_{*})^{\top} P B_{\ru} \eta  \\ 
              & \; {} - (y - Cx + D(\bar{u} + \eta))^{\top} \eta + (x-x_{*})^{\top} P B_{\rv} \nu \\
              &= \dfrac{1}{2} \begin{bmatrix} (x - x_{*})^{\top} & \eta^{\top} \end{bmatrix} 
                  \begin{bmatrix} A^{\top} P + PA & -P B_{\ru} + C^{\top} \\-B_{\ru}^{\top} P + C & -(D + D^{\top}) \end{bmatrix}
                  \begin{bmatrix} x - x_{*} \\ \eta \end{bmatrix} \\
              & \; {} - (y - y_{\rd})^{\top} \eta + (x-x_{*})^{\top} P B_{\rv} \nu \;,
\end{split}
\end{equation}
where in the case of $\nu = 0$ we obtain the energy-balancing 
equation changing the output to $-(y - y_{\rd})$.

\begin{rmark}
Note that the control $\bar{u}$ achieves the asymptotic regulation 
of the output $y$ via a change in the storage function $H_1(x)$
but, once again, $\bar{u}$ is not implementable, as it requires perfect 
knowledge of the state and system parameters and would lead to a closed-loop
system which is not robust.
\end{rmark}

In Section~\ref{subsec:Well} it was established that, when 
$\Phi(y) = \psi_S(y)$ and $D$ is symmetric, we have
\begin{displaymath}
y = \Proj_{S}^{D^{-1}}(C x) \;.
\end{displaymath}
This equation evidently shows the robustness property of the 
multivalued control law $u \in N_S(y)$, since it is not
necessary to maintain the state $x$ at a precise point. Instead, it
is sufficient to maintain $x$ in the set of points for which 
its  projection over $S$ is equal to $y_d$. This result obviously
depends of the shape of the set $S$ and in order to achieve robust
regulation it is necessary that $y_d \in \partial S$ and $\interior
N_S(y_{\rd}) \neq \emptyset$ (see Figures~\ref{fig:phaseProof},~\ref{fig:phaseAnotherSet} below).

The previous argument can be extended for the more general case
where $u \in \partial \Phi(y)$ with $\Phi(y) = (\varphi + 
\psi_S)(y)$ and $\varphi$ an arbitrary proper, convex and lower 
semicontinuous function, i.e. we can achieve robust output 
regulation for  a family of controls parametrized  by $\varphi$.

\begin{theorem}[Main result] \label{th:main}
Consider system~\eqref{eq:sysH1}, and suppose that 
Assumption~\ref{ass:pass} holds. 
Then, the family of controls that satisfy $u \in \partial 
\Phi(y)$, with $\Phi(y) = (\varphi + \psi_S) (y)$ for some
proper, convex and lower semicontinuous function $\varphi$ and 
some closed convex set $S$ specified in the proof, 
yields the robust output regulation $y = y_{\rd}$ in finite 
time whenever 
\begin{equation} \label{eq:ass1}
\langle D^{-1} (y_{\rd}  - Cx_{*}), y_{\rd} \rangle < 
\mathcal{D} \varphi(y_{\rd}, -y_{\rd}) \;, 
\end{equation}
and 
\begin{displaymath}
\| \nu \| \leq \mathcal{B} \;,
\end{displaymath}
for some $\mathcal{B} > 0$ specified along the proof. Here, 
$\mathcal{D} \varphi(y_0, d)$ is the directional derivative 
of the function $\varphi$ at the point $y_0$ in the direction $d$
and $x_{*}$ is the equilibrium associated to the basic IDA design 
(Lemma~\ref{lmma:ubIDA}). Furthermore, in the family of all controls,
there exist at least one that is passive.
\end{theorem}
%


%
\begin{proof}
Applying the control input $u \in \partial \Phi(y)$ 
automatically implies that $y \in S$ (see Subsection~\ref{subsec:Well}).
Then, if we want the regulation of $y$ to $y_{\rd}$ a necessary condition is $y_{\rd} \in S$. 
Consider the following convex set
\begin{equation} \label{eq:Sset}
S = \convex \{ 0, y_{\rd} \}
\end{equation}
where the operator $\convex \{a, b\}$ refers to the convex hull 
of two points $a \in \RE^{m}$ and $b \in \RE^{m} $, i.e.
\begin{displaymath}
\convex \{a , b \} = \{ c \in \RE^{m} : c = \lambda a + 
(1 - \lambda) b, \; \lambda \in [0,1] \} \;.
\end{displaymath}
and consider the following half-space
\begin{displaymath}
\Omega_{\rd} = \{ x \in \RE^{n} : \langle D^{-1} (y_{\rd} - Cx),
y_{\rd} \rangle \leq \mathcal{D}(y_{\rd}, -y_{\rd}) \} \;.
\end{displaymath}
Our first goal is to show that the value of the output $y$
is equal to $y_{\rd}$ whenever $x \in \Omega_{\rd}$.
Assuming $x \in \Omega_{\rd}$ implies
\begin{align*}
\langle D^{-1}(y_{\rd} - Cx), y_{\rd} \rangle & \leq 
\mathcal{D}(y_{\rd}, -y_{\rd}) = \inf_{\rho > 0} 
\dfrac{\varphi(y_{\rd} - \rho y_{\rd}) - \varphi(y_{\rd})}{\rho} \\
& \leq \dfrac{\varphi (\mu y_{\rd}) - \varphi(y_{\rd})}{1 - \mu}
\text{ for all } \mu \in [0,1),
\end{align*}
where we did the change of variables $\rho = 1- \mu$. Because the 
term $1 - \mu$ is positive, we have
\begin{align*}
\langle D^{-1}(y_{\rd} - Cx), (1 - \mu)y_{\rd} \rangle \leq
\varphi(\mu y_{\rd}) - \varphi(y_{\rd}) \text{ for all } \mu \in [0,1]
\end{align*}
Furthermore, each element of $S$ can be represented as
$\sigma = \mu y_{\rd} \in S$ for some $\mu \in [0,1]$, therefore
\begin{displaymath}
\langle D^{-1}(y_{\rd} - Cx), \sigma - y_{\rd} \rangle + 
\varphi(\sigma) - \varphi(y_{\rd}) \geq 0 \text{ for all }
\sigma \in S \;.
\end{displaymath}
That is, $y_{\rd}$  is a solution of the hemivariational 
inequality~\eqref{eq:hemivi} when $x \in \Omega_{\rd}$, and 
considering the uniqueness of solutions, the output $y$ must be 
equal to $y_{\rd}$. 

It remains to show that (even in the presence
of the external perturbation $\nu$), the system state $x$ enters 
the interior of the set $\Omega_{\rd}$ in finite time and 
remains therein for all future time. In terms of the equilibrium
$x_{*}$, we have from~\eqref{eq:ass1} that $x_{*} \in \Omega_{\rd}$.
We will prove that for some $\delta > 0$ small enough, there
exist an ellipsoid $\mathcal{E} = \{ x \in \RE^{n} : 
(x - x_{*})^{\top} P (x - x_{*}) \leq \delta\} \subset \interior
\Omega_{\rd}$ around $x_{*}$ that is attractive and invariant.

%
Considering the dissipation equation~\eqref{eq:eb}, it is clear 
that $\eta = u - \bar{u}$ is well defined, where $\bar{u}$ is the basic
IDA control from Lemma~\ref{lmma:ubIDA}, and $u \in \partial \Phi(y)$.
Then, equation~\eqref{eq:eb} transforms into
\begin{align*}
\dot{H}_2(x) = & \dfrac{1}{2} \begin{bmatrix} (x - x_{*})^{\top} 
& -(u - \bar{u})^{\top} \end{bmatrix} \begin{bmatrix} 
A^{\top} P + PA & P B_{\ru} - C^{\top} \\B_{\ru}^{\top} P - C 
& -(D + D^{\top}) \end{bmatrix} \begin{bmatrix} x - x_{*} \\ 
-(u - \bar{u}) \end{bmatrix} - \\
& (y - y_{\rd})^{\top} (u - \bar{u}) + (x-x_{*})^\top P B_{\rv} \nu,
\end{align*}
where the term $-(y - y_{\rd})^{\top} (u - \bar{u})$ 
is negative for all $y \neq y_{\rd}$ (i.e., for all $x \notin
\Omega_{\rd}$). Indeed, we have from~\eqref{eq:ass1} and Lemma~\ref{lmma:ubIDA} that
\begin{displaymath}
-\langle \bar{u}, y_{\rd} \rangle < \mathcal{D}(y_{\rd}, -y_{\rd}) \;,
\end{displaymath}
and from the definition of subdifferential we have
\begin{align*}
u \in \partial \Phi(y) \Leftrightarrow \varphi(\sigma) - \varphi(y)
\geq \langle u, \sigma - y \rangle \text{ for all } \sigma \in S.
\end{align*}
Specifically, for $\sigma = y_{\rd}$ we obtain
$-\langle u, y - y_{\rd}  \rangle \leq \varphi(y_{\rd}) - \varphi(y)$. Moreover, for all
$y \in S \setminus \{ y_{\rd} \}$ we can write 
$y= \lambda y_{\rd}$ with $\lambda \in [0,1)$. Thus, 
\begin{align*}
-(y - y_{\rd})^{\top}(u - \bar{u}) & \leq \varphi(y_{\rd}) - 
\varphi(\lambda y_{\rd}) - (1 - \lambda) y_{\rd}^{\top} \bar{u} \\
& < \varphi(y_{\rd}) - \varphi(\lambda y_{\rd}) + (1 - \lambda)
\inf_{\mu > 0} \dfrac{\varphi(y_{\rd} - \mu y_{\rd}) -
\varphi(y_{\rd})}{\mu} \\
& \leq  \varphi(y_{\rd}) - \varphi(\lambda y_{\rd}) + (1 - \lambda)
\dfrac{\varphi(y_{\rd} - \mu y_{\rd}) - \varphi(y_{\rd})}{\mu}
 \end{align*}
for all $\mu > 0$. Setting $\mu = 1 - \lambda > 0$ we obtain $-(y - y_{\rd})^{\top}
(u - \bar{u}) < 0$ for all $y \neq y_{\rd}$.

From~\eqref{eq:sysH1:Output} we have that $u$ must satisfy
$u = D^{-1} \left( Cx - y \right)$. Substituting $u$
and $\bar{u}$ in~\eqref{eq:eb} and applying the Lambda 
inequality to the term $(x-x_{*})^{\top} P B_{\rv} \nu$, we have
\begin{align*}
\dot{H}_{2} \leq & -\dfrac{1}{2} w^{\top} R w - 
(y -y_{\rd})^{\top} (u - \bar{u}) + 
(x -x_{*})^{\top} \Lambda (x - x_{*}) + 
\nu^{\top} B_{\rv} P \Lambda^{-1} P B_{\rv} \nu
\end{align*}
where $\Lambda = \Lambda^{\top} > 0$ and
\begin{align*}
w^{\top} & = \begin{bmatrix} (x -x _{^{*}})^{\top} &
(y - y_{\rd})^{\top} D^{-\top} \end{bmatrix} \\
R & = - \begin{bmatrix} A^{\top} P + P A - 
C^{\top} D^{- \top} B_{\ru}^{\top} P - PB_{\ru} D^{-1} C &
PB_{\ru} - C^{\top} D^{- \top} D \\ B_{\ru}^{\top} P - 
D^{\top} D^{-1} C & -(D + D^{\top})
\end{bmatrix}.
\end{align*}
It follows that $-R < 0$, since it is obtained applying the
following non singular congruence transformation
\begin{displaymath}
\begin{bmatrix}
I & 0 \\ -D^{-1} C & I 
\end{bmatrix}
\end{displaymath}
to~\eqref{eq:LMI1}. Now, setting $\Lambda $ in a way that 
\begin{displaymath}
R_{\Lambda} = R - \begin{bmatrix}
\Lambda & 0 \\ 0 & 0 \end{bmatrix} > 0 \;,
\end{displaymath}
we have
\begin{align*}
\dot{H}_2(x) \leq -\dfrac{1}{2} w^{\top} R_{\Lambda} w - 
(y - y_{\rd})^{\top}(u - \bar{u}) + 
\nu^{\top} B_{\rv} P \Lambda^{-1} P B_{\rv} \nu 
\end{align*}
and therefore
\begin{align*}
\dot{H}_2(x) & \leq - \dfrac{1}{2}\lambda_{\min} (R_{\Lambda}) 
\| w \|^{2} + \lambda_{\max} (B_{\rv} P \Lambda^{-1} P B_{\rv}) 
\| \nu \|^{2} \\
& \leq - \dfrac{1}{2} \lambda_{\min} (R_{\Lambda}) \| x - x_{*} 
\|^{2} + \lambda_{\max}(B_{\rv} P \Lambda^{-1} P B_{\rv}) \| 
\nu \|^{2}
\end{align*}
Considering that we are looking for stability of the ellipsoid
$\mathcal{E}$ defined above, we have that, for all $x \notin 
\mathcal{E}$,
\begin{displaymath}
\| x -x_{*} \|^{2} > \dfrac{\delta}{\lambda_{\max}(P)} \;,
\end{displaymath}
and therefore, if $\nu$ satisfies
\begin{displaymath}
\| \nu \|^{2} \leq \dfrac{\delta \lambda_{\min}(R_{\Lambda})}
{2 \lambda_{\max}(P) 
\lambda_{\max}(B_{\rv}^{\top} P \Lambda^{-1} P B_{\rv})} =
\mathcal{B}^{2} \;,
\end{displaymath}
we conclude that $\dot{H}_2 < 0$ for all $x \notin \mathcal{E}$,
i.e. the set $\mathcal{E}$ is attractive and invariant~\cite{khalil2002}.
Now, for passivity of the controller, we have from~\eqref{eq:hemivi} 
that
\begin{displaymath}
\langle u, y \rangle \geq \varphi(y) - \varphi(0) \;.
\end{displaymath}
Therefore, if we choose $\varphi$ such that $\varphi(y) \geq \varphi(0)$ for all $y \in S$,
then the control $u$ will be passive.

Finally, finite-time convergence of the output is obtained
automatically from the proof. Namely, $\mathcal{E} \subset 
\interior \Omega_{\rd}$ together with attractivity and invariance 
of $\mathcal{E}$ implies that there exist a time $t^{*} < \infty$
such that the state will cross the boundary of $\Omega_{\rd}$ and 
will remain inside of $\Omega_{\rd}$ for all $t > t^{*}$.
\end{proof}
Figures~\ref{fig:phaseProof} and~\ref{fig:phaseAnotherSet} show a
picture of the convergence of the term $C x $ to the interior of the
set $\{y_{\rd}\} + N_{S}(y_{\rd})$ in the output space for the sets 
$S = \convex \{0, y_{\rd}\}$ and $S = \convex \{[0,0], [y_{\rd_1}, 0],
[0, y_{\rd_2}], [y_{\rd_1}, y_{\rd_2}] \}$  respectively, for the
case $\varphi = 0$, $D = I_{n}$ and $m = 2$. Note that, $C x - y_{\rd}
\in N_S(y_{\rd})$ is equivalent to $y_{\rd} = \Proj_S(C x)$ and from~\eqref{eq:yProj} we obtain $y = y_{\rd}$.
\begin{figure}[!htb]
\begin{center}
 \includegraphics[width=0.5\textwidth]{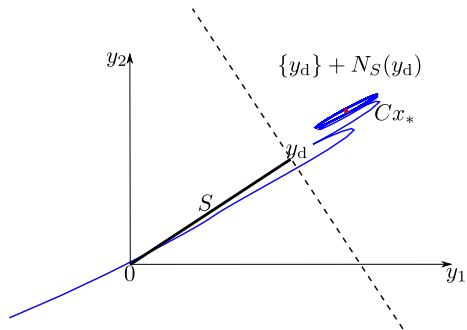}
 \caption{Trajectory of $C x$ converging to interior of $\{y_d\} +
  N_S(y_{\rd})$ with the multivalued control $u \in N_S(y)$ and 
  $S = \convex \{ 0, y_{\rd} \}$. This implies that $y$ converges to $y_\rd$.}
 \label{fig:phaseProof}
\end{center}
\end{figure}
\begin{rmark}
From Figure~\ref{fig:phaseProof} it is possible to see that, if
$x_{*}$ satisfies the condition~\eqref{eq:ass1} for $y_{\rd} \in S$,
then we can achieve robust output regulation for any other desired 
value $\bar{y}_{\rd}$ in the relative interior of $S$ by redefining 
the set $S$  to $\bar{S} = \convex \{0, \bar{y}_{\rd} \}$. Moreover, 
in a more general setting, condition~\eqref{eq:ass1} allows us to 
attack the problem of robust tracking in the following way. Let 
$y_{\rd}(t)$ be the desired reference signal. If, for all values of 
the function $y_{\rd}: \RE \to \RE^{m}$, condition~\eqref{eq:ass1} 
is satisfied together with the bound in $\nu(t)$, then robust output 
tracking is possible as shown in Example~1 below. 
\end{rmark}
\begin{rmark}
It is worth to note that a similar result can be obtained (with
possibly different bounds in the external perturbation and 
different condition in $x_{*}$), if we change the form of the 
set $S$. For example, for $\Phi(y) = \psi_S(y)$ a possible 
set $S$ could be as the one given in Figure~\ref{fig:phaseAnotherSet}, 
where the point $y_{\rd}$ is still in the boundary of $S$ and the
normal cone to $S$ at $y_{\rd}$ has no empty interior, the details 
are left to the reader.
\end{rmark}
\begin{figure}[!htb]
\begin{center}
 \includegraphics[width=0.5\textwidth]{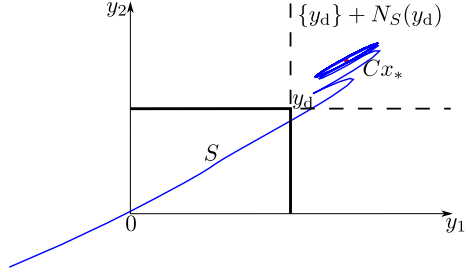}
 \caption{Trajectory $C x$ converging to interior of $\{y_d\} +
  N_S(y_{\rd})$ with the multivalued control $u \in N_S(y)$ and 
  $S = \convex \{[0,0], [y_{\rd_1}, 0],
[0, y_{\rd_2}], [y_{\rd_1}, y_{\rd_2}] \}$.}
 \label{fig:phaseAnotherSet}
\end{center}
\end{figure}

\section{Implementation issues and examples} \label{sec:Implementation}

\subsection{Regularization}
Up to this point, we have shown that whenever $u \in \partial \Phi(y)$
the membership of $y$ to the set $S$, together with robust output
regulation are assured. Our next step is to develop a way to recover 
an explicit expression for the values of the control input $u$ in
terms of the measured output $y$ and independent of the system
parameters and state. 

Note that exact values of input $u$ can be computed by solving 
the hemivariational inequality~\eqref{eq:hemivi} at each time instant
$t$, and making use of~\eqref{eq:sysH1:Output}, but this approach
requires knowledge of the system parameters and state $x$.
  
As another alternative, it is worth noting that the approach of
continuous selections does 
not yield the desired features. For example, in the 
case of $\varphi = 0$, a continuous selection of the multifunction 
$\partial \Phi(\cdot) = N_S(\cdot)$ is $u = 0$, 
(in fact $u = 0$ is the unique continuous selection). However,
with that control the storage function of~\eqref{eq:sysH1} is given
by $H_1$ with minimum at $\bar{x}$ and consequently, neither robust
output regulation nor $y \in S$ properties are (in general) 
obtained. Similar results can be obtained when $\varphi \in 
\mathcal{C}^{1}$, since $u = \nabla \varphi (\cdot)$ is always 
a continuous selection of $\partial \Phi(\cdot)$.

Instead of looking for continuous selections of $\partial
\Phi(\cdot)$, we are going to focus on a regularization of $\graph
\partial \Phi$ in the sense used in~\cite{miranda2014}. More
precisely,
\begin{equation} \label{eq:regControlMult}
\tilde{u} - \nabla \varphi(y) \in N_S(y - 
\varepsilon \left[\tilde{u} - \nabla \varphi (y) \right] )
\end{equation}
is a regularization of the inclusion $u \in \partial \Phi(y)$.
Namely, note that for $\varepsilon = 0$ we recover
$\tilde{u} \in \partial \Phi(y)$ (because $\partial \Phi(y) = 
\nabla \varphi (y) + N_S(y)$). Moreover, with the previous 
definition we are allowing outputs $y$ not necessarily in $S$.
Instead we now require $y \in \{ \varepsilon \left[ \tilde{u} - 
\nabla \varphi (y) \right] \} + S$.

The well-posedness of inclusion~\eqref{eq:regControlMult} together 
with a single valued expression for $\tilde{u}$ are established below 
in Theorem~\ref{th:reg}. The following Lemma will be useful when proving it.

\begin{lemma} \label{lma:Contraction}
The map $f: \RE^{m} \to \RE^{m}$ given by
\begin{displaymath}
f(z) := (I + \varepsilon D^{-1})^{-1} z \;,
\end{displaymath}
where $D^{-1} + D^{-T} > 0$, is a contraction for all $\varepsilon > 0$.

\end{lemma}
\begin{proof}
Defining 
\begin{displaymath}
\zeta = f(z) \;,
\end{displaymath}
we have $(I + \varepsilon D^{-1}) \zeta = z$ and direct computation gives
\begin{align*}
\| z_1 - z_2 \|^{2} &= \left[ (I + \varepsilon D^{-1})(\zeta_1 - \zeta_2) \right]^{\top} \left[ (I + \varepsilon D^{-1})(\zeta_1 - \zeta_2) \right] \\
                    &= \| \zeta_1 - \zeta_2 \|^{2} + \varepsilon (\zeta_1 - \zeta_2)^{\top}\left[ D^{-1} + D^{-\top} \right] (\zeta_1 - \zeta_2) \\ 
                     & \: {} + \varepsilon^{2}(\zeta_1 - \zeta_2)^{\top} D^{-\top}D^{-1} (\zeta_1 - \zeta_2) \\
                    &\geq \| \zeta_1 - \zeta_2 \|^{2} + \varepsilon \lambda_{\min} \left( D^{-1} + D^{-\top} \right) \| \zeta_1 - \zeta_2 \|^{2} \\
                     & \: {} + \varepsilon^{2} \lambda_{\min} \left( D^{-\top} D^{-1} \right) \| \zeta_1 - \zeta_2 \|^{2} \;.
\end{align*}
Therefore,
\begin{displaymath}
\| f(z_1) - f(z_2) \| \leq \dfrac{1}{\sqrt{1 + \varepsilon
\lambda_{\min} \left( D^{-1} + D^{-\top} \right) + \varepsilon^{2} 
\lambda_{\min} \left( D^{-\top} D^{-1} \right)}} \| z_1 - z_2 \| \;.
\end{displaymath}
\end{proof}

\begin{theorem} \label{th:reg}
Let $\varphi$ be a strictly convex, lower semicontinuous function  
that is $\mathcal{C}^{1}$ and satisfies
\begin{enumerate}
\item[•] $\varphi(y) \geq \varphi (0)$ for all $y \in S$.

\item[•]  $\nabla \varphi : \RE^{m}  \to \RE^{m}$ is Lipschitz continuous with constant $L$ such that
\begin{displaymath}
L < \lambda_{\min} \left( \dfrac{ D^{-1} + D^{-\top} } {2} \right) \;.
\end{displaymath} 
\end{enumerate}
Then, for $\varepsilon > 0$ sufficiently small, the regularized
control $\tilde{u}$ can be expressed as:
\begin{equation} \label{eq:regControlSing}
\tilde{u} = \dfrac{y - \Proj_S(y)}{\varepsilon} + \nabla \varphi
(y)
\end{equation}
Furthermore, $\tilde{u}$ is passive respect to $y$.
\end{theorem}

\begin{proof}
From~\eqref{eq:regControlMult} we have that for all $\sigma \in S$ the 
following holds:
\begin{equation} \label{eq:regVI}
0 \geq \langle \tilde{u} - \nabla \varphi(y), \sigma - y +
\varepsilon \left[ \tilde{u} - \nabla \varphi (y) \right] \rangle \;.
\end{equation}
Multiplying by $\varepsilon > 0 $ and adding and subtracting $y$ on
the left-hand side of the inner product we obtain
\begin{displaymath}
0 \geq \langle y - y + \varepsilon \left[ \tilde{u} - \nabla
\varphi (y) \right], \sigma - y + \varepsilon \left[ \tilde{u} -
\nabla \varphi (y) \right] \;.
\end{displaymath}
Therefore,
\begin{displaymath}
y - \varepsilon \left[ \tilde{u} - \nabla \varphi (y) \right]
= \Proj_S(y) \;,
\end{displaymath}
from which we obtain~\eqref{eq:regControlSing}. Now we show that
the interconnection of the plant~\eqref{eq:sysH1:Dyn}--\eqref{eq:sysH1:Output} 
with the regularized control~\eqref{eq:regControlSing} is well-posed. It is easy to see that 
well-posedness of the closed-loop system is equivalent to proving 
that, for any state $x \in \RE^{n}$, the equations
\begin{align*}
\tilde{u} &  = \dfrac{y - \Proj_S(y)}{\varepsilon} + \nabla 
\varphi (y), \\
\tilde{u} & = D^{-1} \left( C x - y \right),
\end{align*}
have a unique solution. Proceeding with the substitution of the 
second equation and after some manipulations we have
\begin{displaymath}
y = \left( I + \varepsilon D^{-1} \right)^{-1} \left[ \Proj_S(y)
- \varepsilon \nabla \varphi (y) + \varepsilon D^{-1} C x \right] 
=  \left( f \circ g \right)  \left( y \right) \;,
\end{displaymath}
with $f$ as in Lemma~\ref{lma:Contraction} and $g: \RE^{m} \to \RE^{m}$
given by
\begin{displaymath}
g(z) = \Proj_S(z) - \varepsilon \nabla \varphi (z) + \varepsilon D^{-1}
C x \;.
\end{displaymath}
We argue that the composition mapping $f \circ g$ is a contraction for
$\varepsilon$ sufficiently small. Indeed, making use of Lemma~\ref{lma:Contraction} we have that
\begin{align*}
\big\| \left( f \circ g \right) \left( y_1 \right) - 
\left(f \circ g \right) \left( y_2 \right) \big\| \leq \dfrac{1}
{\beta(\varepsilon)} \|g(y_1) - g(y_2) \| \leq \dfrac{1 + \varepsilon L}
{\beta(\varepsilon)} \| y_1 - y_2 \|
\end{align*}
where $\beta(\varepsilon) = \left[ 1 + \varepsilon \lambda_{\min} 
\left( D^{-1} + D^{-\top} \right) + \varepsilon^{2} \lambda_{\min} 
\left( D^{-\top} D^{-1} \right) \right]^{1/2}$. Note that the term 
$ \frac{1 + \varepsilon L}{\beta(\varepsilon)}$ is equal to $1$
for $\varepsilon = 0$ and 
\begin{displaymath}
\dfrac{d}{d \varepsilon} \left( \dfrac{1 + \varepsilon L}
{\beta(\varepsilon)} \right) \bigg|_{\varepsilon = 0}
= L - \lambda_{\min} \left( \dfrac{ D^{-1} + D^{-\top}}{2} \right) < 0 \;,
\end{displaymath}
i.e. the term $\frac{1 + \varepsilon L}{\beta(\varepsilon)}$ is
strictly decreasing in a neighbourhood of $\varepsilon = 0$ and 
thus it is less than $1$ for $\varepsilon$ sufficiently small.
Therefore, $f \circ g$ is a contraction and the interconnection is
well-posed. It only rests to prove the passivity property of 
$\tilde{u}$. From~\eqref{eq:regVI} we have for $\sigma = 0 \in S$
\begin{displaymath}
\langle \tilde{u}, y \rangle \geq \langle \nabla 
\varphi (y), y \rangle + \varepsilon \| \tilde{u} - \nabla 
\varphi (y) \|^{2} \;.
\end{displaymath}
Note that for $\varepsilon = 0$ we have $y \in S$ and
from the strictly convexity assumption (see e.g.~\cite[p. 183]{hiriart1993}),
\begin{displaymath}
\langle \nabla \varphi (y), y \rangle > \varphi(y) - \varphi(0) \geq 0 \text{ for all } y \in S \;.
\end{displaymath}
In other words we have
\begin{displaymath}
\lim_{\varepsilon \downarrow 0} \langle \nabla 
\varphi (y), y \rangle + \varepsilon \| \tilde{u} - \nabla 
\varphi (y) \|^{2} > 0 \;.
\end{displaymath}
Consequently, $\langle \tilde{u}, y \rangle \geq 0$ for some 
$\varepsilon > 0$ sufficiently small. 
\end{proof}

\begin{rmark}
Note that Theorem~\ref{th:reg} is still true if we change the first
assumption by $\varphi(y) \geq \varphi(0) $ for all $y \in \Dom \varphi$ with $\varphi$ a convex function.
\end{rmark}

\subsection{Example 1}

Consider the circuit described by the diagram of
Figure~\ref{fig:ExCirc_Circ}. We wish to regulate the outputs $y_1$ 
and $y_2$ to a desired value $y_d$. 
\begin{figure}[!hb]
\begin{center}
 \includegraphics[width=0.25\textwidth]{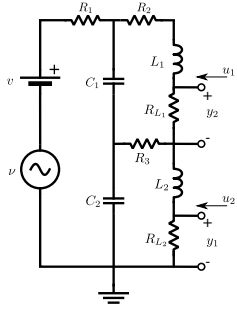}
 \caption{Circuit diagram of Example~1, where the goal is
  to regulate the voltage at the outputs $y_1$ and $y_2$.}
 \label{fig:ExCirc_Circ}
\end{center}
\end{figure}
Taking as state variables the fluxes in inductors and charges in 
capacitors, we have the following state-space representation:
\begin{subequations} \label{eq:Circuit}
\begin{align}
\dot{x} & = \begin{bmatrix}
- \frac{1}{R_1 C_1} & -\frac{1}{L_1} & - \frac{1}{R_1 C_2} & 0 \\
\frac{1}{C_1} & - \frac{R_2 + R_3 + R_{L_1}}{L_1} & 0 & 
\frac{R_3}{L_2} \\
-\frac{1}{R_1 C_1} & 0 & -\frac{1}{R_1 C_2} & -\frac{1}{L_2} \\
0 & \frac{R_3}{L_1} & \frac{1}{C_2} & -\frac{R_3 + R_{L_2}}{L_2}
\end{bmatrix} x + \begin{bmatrix}
0 & 0 \\
R_{L_1} & 0 \\
0 & 0 \\
0 & R_{L_2}
\end{bmatrix} u + \begin{bmatrix}
\frac{1}{R_1} \\ 0 \\ \frac{1}{R_1} \\ 0 
\end{bmatrix} v \\
y & = \begin{bmatrix}
0 & \frac{R_{L_1}}{L_1} & 0 & 0 \\
0 & 0 & 0 & \frac{R_{L_2}}{L_2}
\end{bmatrix} x + \begin{bmatrix}
R_{L_1} & 0 \\
0 & R_{L_2}
\end{bmatrix} u
\end{align}
\end{subequations}
where $x = \begin{bmatrix} x_1 & x_2 & x_3 & x_4 \end{bmatrix}^{\top}$
are the charge in capacitor $C_1$, flux in inductor $L_1$, charge in
capacitor $C_2$ and flux in inductor $L_2$, respectively,
$u = \begin{bmatrix} u_1 & u_2 \end{bmatrix}^{\top}$ are the control
inputs (currents) and $y = \begin{bmatrix} y_1 & y_2 \end{bmatrix}
^{\top}$ are the voltages in resistances $R_{L_1}$ and $R_{L_2}$,
respectively. Assume that we want to control the outputs 
to $y_d = \begin{bmatrix} 1 & f(t) \end{bmatrix}^{\top}$, where 
$f(t)$ is a sawtooth wave function with amplitude $0.5$ and frequency
of $2$ Hz.

The system is passive because it is the result of the interconnection 
of passive elements. Values of system parameters are $R_1 = R_2 =
R_3 = 1 \Omega$, $R_{L_1} = 2 \Omega$, $R_{L_2} = 3 \Omega$, $L_1 =
1$H, $L_2 = 2$H, $C_1 = 1$F, $C_2 = 3$F, $v = 10 + 50 \sin(t) 
\sign(\sin( \pi t))$. Taking the convex function $\varphi = 0$, 
simple algebra shows that condition~\eqref{eq:ass1} is equal to
\begin{displaymath}
 \langle D^{-1} (y_{\rd} - C x_{*}), y_{\rd} \rangle = -4 -4 f(t)
 + \dfrac{5}{6} f(t)^{2} \;,
\end{displaymath}
which is negative for values of $f(t) \in (-0.849, 5.649)$. The
implemented control takes the form~\eqref{eq:regControlSing} with
$\varepsilon = 1 \times 10^{-3}$ and $S$ the convex, time-varying set
\begin{displaymath}
S(t) = \convex \left\lbrace \begin{bmatrix} 0 \\ 0 \end{bmatrix},
\begin{bmatrix} 1 \\ f(t) \end{bmatrix} \right\rbrace \;.
\end{displaymath}

Figure~\ref{fig:ExCirc_y} shows the convergence of the output to the 
desired reference, even in the presence of the external perturbation 
$v$. Moreover, is easy to see that the condition $y \in S$ is
satisfied. The computed control input is shown in Figure~\ref{fig:ExCirc_u}.

%
\begin{figure}[!htb]
\begin{center}
 \includegraphics[width=1\textwidth]{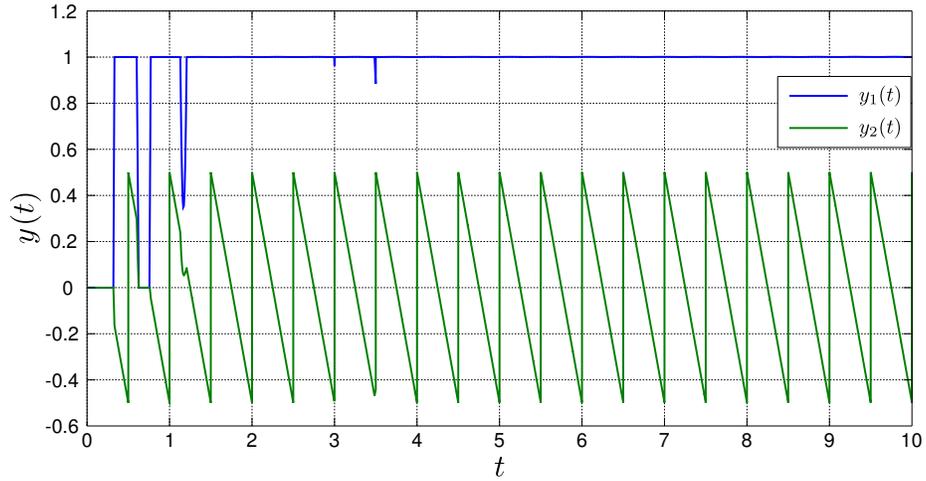}
 \caption{Output response of plant~\eqref{eq:Circuit} with
 regularized control~\eqref{eq:regControlSing}. The picture shows  
 convergence to the desired output $y_{\rd}$ (with $\varepsilon = 1
 \times 10^{-3}$ and $\varphi = 0$) subject to the perturbation 
 $v(t) = 10 + 50 \sin(t) \sign(\sin( \pi t)) $.}
 \label{fig:ExCirc_y}
\end{center}
\end{figure}
\begin{figure}[!htb]
\begin{center}
 \includegraphics[width=1\textwidth]{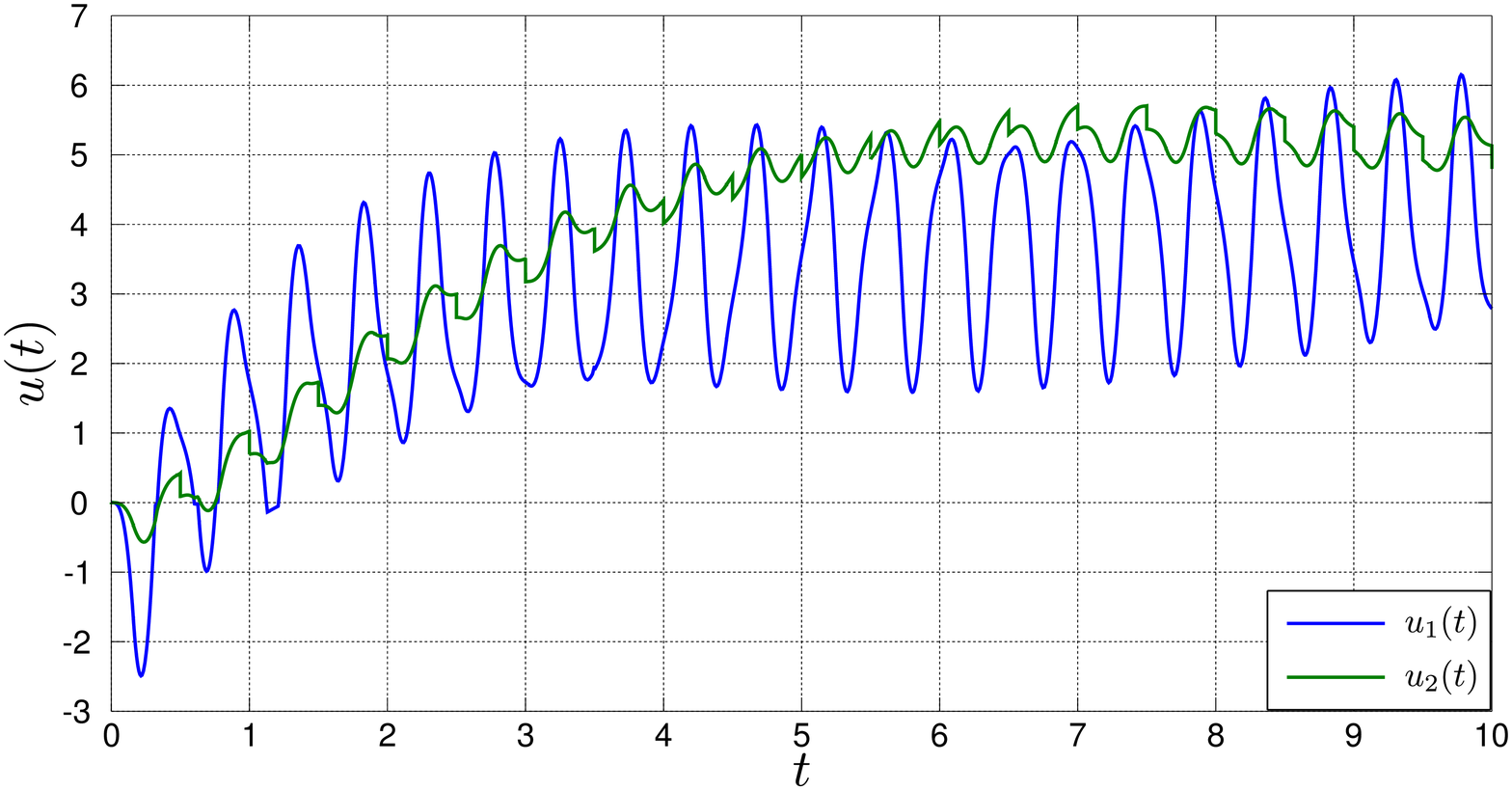}
 \caption{Time trajectory of regularized control~\eqref{eq:regControlSing} 
  with $\varepsilon = 1 \times 10^{-3}$ and
 $\varphi = 0$.}
 \label{fig:ExCirc_u}
\end{center}
\end{figure}

\subsection{Example 2}

Consider the following affine system
\begin{subequations} \label{eq:AbstactPlant}
\begin{align}
\dot{x} &= A x + B_{\ru} u_1 + B_{\rv} v \\
y_1 &= Cx + Du_1
\end{align}
\end{subequations}
with
\begin{align*}
A & = \begin{bmatrix}
-3.7036 & 1.9043 & -0.9735 & 0.4164 \\
-0.2421 & -5.1187 & -0.0478 & 0.0269 \\
-0.9915 & -1.0461 & -5.5232 & -0.5318 \\
0.4376 & 2.1467 & -0.5948 & -3.6545
\end{bmatrix}, 
& B_{\ru} = & \begin{bmatrix}
-0.6918 & -1.4410 \\
0.8580 & 0.5711 \\
1.2540 & -0.3999 \\
-1.5937 & 0.6900
\end{bmatrix}, \\
C & = \begin{bmatrix}
0.0652 & 0.4889 & 0.6820 & 0.9166 \\
0.7134 & 0.6677 & 0.1996 & 0.8659
\end{bmatrix}, 
& B_{\rv} = & \begin{bmatrix}
0.8147 & 1.4345 \\
0.9058 & 2.5464 \\
-0.1270 & 0 \\
0.9134 & -1.0453 \\
\end{bmatrix}, \\
D & = \begin{bmatrix}
1.0823 & 0.3899 \\ 
-0.1315 & 0.088
\end{bmatrix} \;.
\end{align*}
where the external perturbation signal $v(t)$ is decomposed
as
\begin{equation} \label{eq:disturbance}
v(t) = \begin{bmatrix} 4 \\ 0 \end{bmatrix} + \begin{bmatrix}
f_1(t) \\ f_2(t) \end{bmatrix}
\end{equation}
with $f_1(t)$ a sinusoidal function with amplitude $2$ and
frequency of 10 Hz and $f_2(t)$ corresponds to a sawtooth wave
with amplitude $3$ and frequency of $\pi$ Hz. Suppose that we want 
to regulate the output to the set-point $y_{\rd} = \begin{bmatrix} -1 
& 2 \end{bmatrix}$.

Let us verify the assumptions of Theorem~\ref{th:main}. The 
equilibrium point $x_{*}$ is 
\begin{displaymath}
 x_{*} = \begin{bmatrix}  2.7809 & 0.1184 & -0.2779 &  0.4877
\end{bmatrix}^{\top} \;.
\end{displaymath}
and it satisfies
\begin{displaymath}
\langle D^{-1} (y_{\rd} - C x_{*}), y_{\rd} \rangle = -9.2810 \;.
\end{displaymath}
Taking, for example, the convex function $\varphi(y) = \log \left( 
e^{y_1} + e^{y_2} \right)$, which is proper and $\mathcal{C}^{1}$, 
we have that
\begin{displaymath}
\mathcal{D} \varphi (y_{\rd}, -y_{\rd}) = - \langle \nabla \varphi
(y_{\rd}), y_{\rd} \rangle = -1.8577 \;.
\end{displaymath}
Condition~\eqref{eq:ass1} is satisfied. Using the SDPT3 software 
to solve~\eqref{eq:LMI1} we obtain
\begin{displaymath}
P = \begin{bmatrix}
1.8765 & 1.8706 & -0.5249 & 1.3338 \\
1.8706 & 3.8984 & -0.4599 & 0.9207 \\
-0.5249 & -0.4599 & 2.4211 & 0.4920 \\
1.3338 & 0.9207 & 0.4920 & 2.0056
\end{bmatrix} \;,
\end{displaymath}
which is positive definite with eigenvalues in $\{ 0.2296, 1.5399,
2.7431, 5.6889 \}$. Figure~\ref{fig:Ex1_y} shows the output response
for a regularized control $\tilde{u}$ with $\varepsilon = 1 \times
10^{-4}$, where finite time convergence toward the desired set-point
can be verified despite the external parametric disturbances of the 
system. Control and state trajectories are shown in Figures~\ref{fig:Ex1_u}
and~\ref{fig:Ex1_x}, respectively.
\begin{figure}[!htb]
\begin{center}
 \includegraphics[width=1\textwidth]{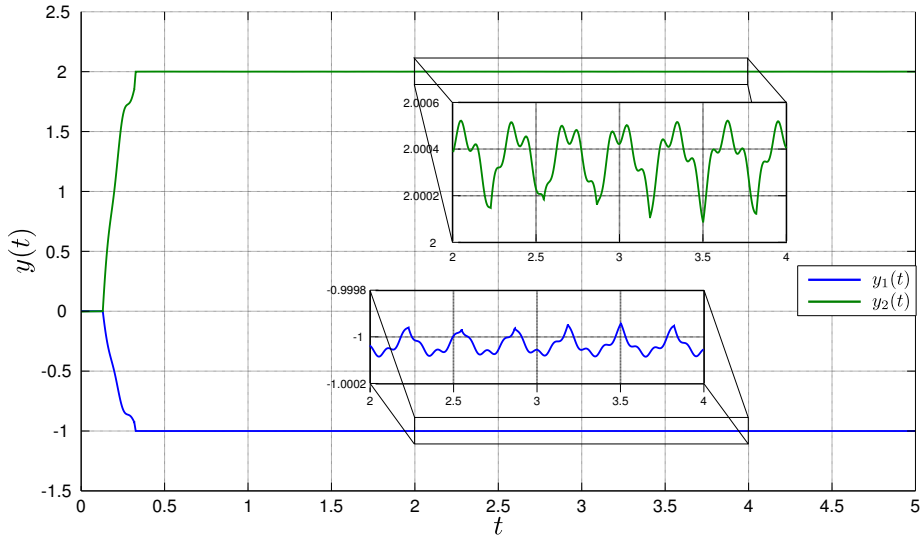}
 \caption{Output's time trajectory for plant~\eqref{eq:AbstactPlant} 
 showing the convergence to the desired value $y_{\rd} = [-1 \quad
  2]^{\top}  $.}
  \label{fig:Ex1_y}
\end{center}
\end{figure}
\begin{figure}[!htb]
\begin{center}
 \includegraphics[width=1\textwidth]{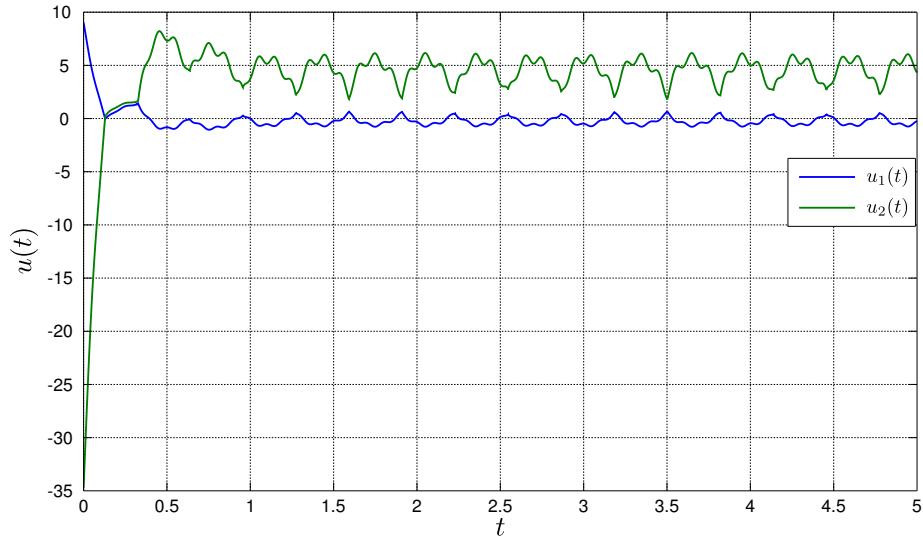}
 \caption{Time trajectory of regularized control~\eqref{eq:regControlSing}
  applied to the plant~\eqref{eq:AbstactPlant} with $\varepsilon = 1 \times 10^{-4}$
  and $\varphi = \log (e^{y_1} + e^{y_2})$.}
 \label{fig:Ex1_u}
\end{center}
\end{figure}
\begin{figure}[!htb]
\begin{center}
 \includegraphics[width=1\textwidth]{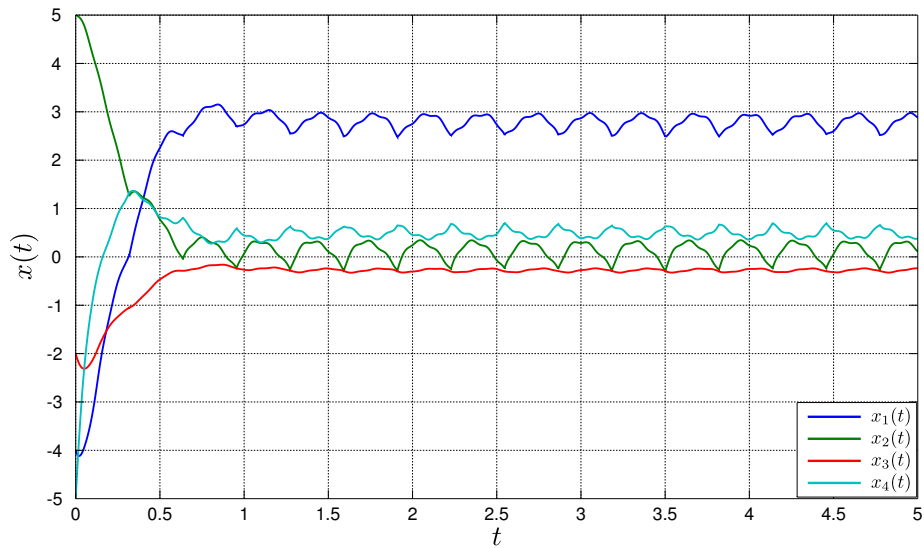}
 \caption{State's time trajectories of regulated plant~\eqref{eq:AbstactPlant} 
  with the regularized control~\eqref{eq:regControlSing} (with
  $\varepsilon = 1 \times 10^{-4}$ and $\varphi = \log (e^{y_1} +  e^{y_2})$) and
  perturbation~\eqref{eq:disturbance}.}
 \label{fig:Ex1_x}
\end{center}
\end{figure}

\section{Conclusions} \label{sec:Conclusions}

This note presents an extension (for the $m$-dimensional case) of the
multivalued control presented in~\cite{miranda2014}. Moreover, more
general multivalued functions of the form $u \in \partial \Phi(y)$ 
are considered, assuring finite time convergence together with,
robust output regulation in the face of parametric and
external (bounded) disturbances.

The effect of the multivalued control relies directly on the
dissipation term modifying the rate of convergence of the storage
function $H_2$ to $x_{*}$ and leaving without change the 
interconnection matrix $J$.

Between the main assumptions considered, the fact that $D$ is 
invertible plays an essential role. A research line is the case of no
$D$ (i.e. $y = C x$).

The implemented control~\eqref{eq:regControlSing} acts in fact as a
high gain controller when $y \notin S$ and coincides with the
continuous selection of $\partial \Phi(y)$ when $y \in S$. However,
since the output contains a feedthrough component of the input, the high gain
does not result in arbitrary large controls. That is, the control converges to
a bounded, well-defined value as $\varepsilon \to 0$. It is 
worth noting that the resulting controller is passive and independent
of the system parameters and of the system state.

The well-suited structure of Port-Hamiltonian systems together with
passivity opens the opportunity to investigate the robust output
regulation problem in the nonlinear setting.

\bibliographystyle{plain}
\bibliography{RobReg}

\end{document}